\documentclass[graybox]{svmult}

\usepackage{url}
\usepackage{type1cm}        % activate if the above 3 fonts are
\usepackage{makeidx}         % allows index generation
\usepackage{graphicx}        % standard LaTeX graphics tool
\usepackage{multicol}        % used for the two-column index
\usepackage[bottom]{footmisc}% places footnotes at page bottom

\usepackage{newtxtext}       % 

\usepackage{enumerate}
\usepackage{cite}
\usepackage{textcomp}
\usepackage{amsmath, amssymb, amsfonts}
\usepackage{mathrsfs}
\usepackage{mathtools}
\usepackage{dsfont}
\usepackage{accents}
\usepackage{color}
\usepackage{caption}

\DeclareMathOperator{\minimize}{minimize}

\DeclareMathOperator{\diag}{diag}
\DeclareMathOperator{\voi}{VoI}
\DeclareMathOperator{\EXP}{\mathsf{E}}
\DeclareMathOperator{\Cov}{\mathsf{cov}}

\DeclareMathOperator{\ProbM}{\mathsf{P}}
\DeclareMathOperator{\Prob}{\mathsf{p}}

\makeindex             % used for the subject index
\begin{document}

\renewcommand\footnotemark{}
\title*{Consistency of Value of Information: Effects\\of Packet Loss and Time Delay in Networked Control Systems Tasks \thanks{\hspace{-4.2mm}Corresponding Author: Touraj Soleymani (touraj@imperial.ac.uk). To be published as a book chapter by \emph{Springer}.}}
\titlerunning{Consistency of Value of Information: Effects of Packet Loss and Time Delay}
\author{Touraj Soleymani, John S. Baras, Siyi Wang, Sandra Hirche, and Karl H. Johansson}
\institute{Touraj Soleymani \at University of London, United Kingdom
\and John S. Baras \at University of Maryland, United States
\and Siyi Wang \at Technical University of Munich, Germany
\and Sandra Hirche \at Technical University of Munich, Germany
\and Karl H. Johansson \at Royal Institute of Technology, Sweden}
\maketitle

\vspace{-10mm}
\abstract{In this chapter, we study the consistency of the value of information---a semantic metric that claims to determine the right piece of information in networked control systems tasks---in a lossy and delayed communication regime. Our analysis begins with a focus on state estimation, and subsequently extends to feedback control. To that end, we make a causal tradeoff between the packet rate and the mean square error. Associated with this tradeoff, we demonstrate the existence of an optimal policy profile, comprising a symmetric threshold scheduling policy based on the value of information for the encoder and a non-Gaussian linear estimation policy for the decoder. Our structural results assert that the scheduling policy is expressible in terms of $3d-1$ variables related to the source and the channel, where $d$ is the time delay, and that the estimation policy incorporates no residual related to signaling. We then construct an optimal control policy by exploiting the separation principle.}

\section{Introduction}
Packet loss and time delay are two common network imperfections in networked control systems tasks that can severely degrade the overall system performance or even yield instability. For that reason, it is imperative that any semantics-aware methodology~\cite{uysal2022semantic}, developed to integrate the communication purpose into the design process, possesses the capability to counteract the detrimental effects of these imperfections. In this chapter, we study the consistency of the value of information~\cite{voi, voi2, touraj-thesis, soleymaniCUP}---a semantic metric that claims to determine the right piece of information in networked control systems tasks---in a lossy and delayed communication regime. Our analysis begins with a focus on state estimation, and subsequently extends to feedback control. This modularized exposition, in particular, enables us to illustrate the applicability of the value of information to remote monitoring. Note that remote monitoring, in which sensory measurements of a stochastic source are transmitted in real-time by an encoder to a decoder that estimates the state of the source, is the most basic task arising in the context of networked control systems~\cite{hespanha2007survey, nair2007feedback, touraj-power}.

The base of our study is a causal tradeoff between the packet rate and the mean square error, which we formulate as a stochastic optimization problem with an encoder and a decoder as two distributed decision makers. This problem for the joint design of the encoder and the decoder is a team decision-making problem, with a non-classical information structure subject to a signaling effect, that is nonconvex and in general intractable. Despite these perplexities, our main objective is to characterize a globally optimal policy profile composed of a scheduling policy and an estimation policy to be used by the encoder and the decoder, respectively. We then attempt to construct an optimal control policy by exploiting the separation principle. Analogous to the original findings in~\cite{voi, voi2, touraj-thesis}, we will quantify the value of information at each time as the variation in a value function defined from the perspective of the encoder with respect to a piece of sensory measurement that can be communicated to the decoder at that~time.

\subsection{Literature Survey}
The causal tradeoff between the packet rate and the mean square error has been examined in the literature~\cite{imer2010, lipsa2011, lipsa2009optimal, molin2017, chakravorty2016, chak2016loss, chak2017loss, rabi2012, guo2021-IT, guo2021-TAC, sun2019}. Notably, Imer and Basar \cite{imer2010} studied the estimation of a scalar Gauss--Markov process based on dynamic programming by considering a symmetric threshold scheduling policy, and obtained the optimal threshold value. Lipsa and Martins~\cite{lipsa2011} analyzed the estimation of a scalar Gauss--Markov process based on majorization theory, and proved that the optimal scheduling policy is symmetric and the optimal estimation policy is linear. This work was extended to estimation over an independent and identically distributed (i.i.d.) packet-erasure channel~in \cite{lipsa2009optimal}, where a similar structural result was found. In addition, Molin and Hirche~\cite{molin2017} investigated the convergence properties of an iterative algorithm for the estimation of a scalar Markov process with symmetric noise distribution, and found a result coinciding with that in~\cite{lipsa2011}. Chakravorty and Mahajan~\cite{chakravorty2016} studied the estimation of a scalar autoregressive Markov process with symmetric noise distribution based on renewal theory, and proved that the optimal scheduling policy is symmetric and the optimal estimation policy is linear. This work was generalized to estimation over an i.i.d.~packet-erasure channel and a Gilbert--Elliott packet-erasure channel~in \cite{chak2016loss,chak2017loss}, where a similar structural result was found. Moreover, Rabi~\emph{et~al.}~\cite{rabi2012} formulated the estimation of the scalar Wiener process and a scalar Ornstein--Uhlenbeck process as an optimal multiple stopping time problem by discarding the signaling effect, and showed that the optimal scheduling policy is symmetric. Later, Guo and Kostina \cite{guo2021-IT} made a contribution by studying the estimation of the scalar Wiener process and a scalar Ornstein--Uhlenbeck process in the presence of the signaling effect, obtained a result that matches with that in~\cite{rabi2012}, and showed that, under some assumptions, a sign-of-innovation code can be used as the optimal compression policy. The authors also looked at the estimation of the scalar Wiener process over a fixed-delay lossless channel in the presence of the signaling effect in~\cite{guo2021-TAC}, and obtained similar structural results. Finally, Sun~\emph{et~al.}~\cite{sun2019} studied the estimation of the scalar Wiener process over a random-delay lossless channel by discarding the signaling effect, and showed that the optimal scheduling policy is symmetric. Note that, in these studies, the scheduling policies are observation-based, as they take advantage of realized sensory information.

There is a related body of research in the literature about the effects of packet loss and time delay on stability of estimation~\cite{sinopoli, wu2017, quevedo2013, schenato2008delay, gupta2009d}. Notably, in a seminal work, Sinopoli~\emph{et~al.}~\cite{sinopoli} studied mean-square stability of Kalman filtering over an i.i.d.~packet-erasure channel, and proved that there exists a critical point for the packet-loss probability above which the expected estimation error covariance is unbounded. Wu~\emph{et~al.}~\cite{wu2017} addressed peak-covariance stability of Kalman filtering over a Gilbert-Elliott packet-erasure channel, and proved that there exists a critical region defined by the recovery and failure rates outside which the expected prediction error covariance is unbounded. Quevedo~\emph{et~al.}~\cite{quevedo2013} investigated mean-square stability of Kalman filtering over a fading packet-erasure channel with correlated gains, and established a sufficient condition that ensures the exponential boundedness of the expected estimation error covariance. Schenato~\cite{schenato2008delay} studied mean-square stability of Kalman filtering over a random-delay i.i.d.~packet-erasure channel, obtained the optimal structure of the estimator at the decoder, and showed that there exists a critical value for the packet-loss probability, similar to what was found in~\cite{sinopoli}, that is independent of the time delay. Note that, in these studies, it is assumed that raw measurements of the sensor are periodically transmitted by the encoder to the decoder. Differently, Gupta~\emph{et~al.}~\cite{gupta2009d} investigated the estimation of a Gauss-Markov process over a packet-erasure channel as a subproblem, and showed that transmitting the minimum mean-square-error state estimate at the encoder at each time leads to the maximal information set for the decoder. The authors also obtained a necessary and sufficient condition for the packet-loss probability of an i.i.d.~packet-erasure channel that guarantees the boundedness of the expected estimation error~covariance.

\subsection{Contributions and Organization}
In this chapter, based on the results established in~\cite{erasure2023}, we first study the fundamental performance limit of state estimation of a partially observable Gauss--Markov process over a fixed-delay packet-erasure channel, by making a causal tradeoff. Associated with this tradeoff, we demonstrate the existence of an optimal policy profile, comprising a symmetric threshold scheduling policy based on the value of information for the encoder and a non-Gaussian linear estimation policy for the decoder. Our structural results assert that the scheduling policy is expressible in terms of $3d-1$ variables related to the source and the channel, where $d$ is the time delay, and that the estimation policy incorporates no residual related to signaling. Then, by adopting the results in~\cite{voi,voi2}, we show that the optimal control policy in the corresponding feedback control problem is a certainty-equivalent policy that exploits the derived estimation policy. Finally, through a numerical example, we demonstrate that the optimal scheduling policy based on the value of information not only transmits measurements less frequently when the estimation discrepancy is small, but transmits more frequently and more persistently when the estimation discrepancy is large due to packet losses and time~delay.

We should remark that our study is different from the previous studies on state estimation of stochastic sources over communication channels~\cite{imer2010, lipsa2011, lipsa2009optimal, molin2017, chakravorty2016, chak2016loss, chak2017loss, rabi2012, guo2021-IT, guo2021-TAC, sun2019,kushner, meier1967, mywodespaper, soleymani2016-cdc, leong2017, leong2018, witsenhausen1979, walrand1983, yuksel2012, borkar2001, khina2018t, tanaka2016,sinopoli, wu2017, quevedo2013, schenato2008delay, gupta2009d}. In particular, it differs from the studies in~\cite{imer2010, lipsa2011, molin2017, chakravorty2016, rabi2012, guo2021-IT}, which consider only ideal channels; from those in~\cite{lipsa2009optimal, chak2016loss, chak2017loss}, which are restricted to scalar sources and delay-free packet-erasure channels; from those in \cite{sun2019, guo2021-TAC}, which are restricted to scalar sources and delayed lossless channels; from those in~\cite{kushner, meier1967, mywodespaper, soleymani2016-cdc,leong2017, leong2018}, which are restricted to variance-based scheduling policies and delay-free packet-erasure channels; and from those in \cite{witsenhausen1979, walrand1983, borkar2001, yuksel2012}, which focus on compression policies rather than scheduling policies. Our results here apply to multi-dimensional sources and fixed-delay packet-erasure channels, without any limitations on the information structure or the policy structure, and deliver an observation-based scheduling policy and a recursive estimation policy that are jointly optimal. Moreover, in contrast to the results in~\cite{sinopoli, wu2017, quevedo2013, schenato2008delay, gupta2009d}, which provide conditions guaranteeing stability of estimation, our results provide a basis for obtaining the minimum packet rate that guarantees a given level of estimation performance. Lastly, contrary to the studies in~\cite{sinopoli, wu2017, quevedo2013, schenato2008delay}, where the message that is transmitted at each time is the output of the sensor at that time, in our study, the message that can be transmitted at each time, similar to what was adopted in \cite{gupta2009d, leong2017}, is the minimum mean-square-error state estimate at the encoder at that time. As we will see, transmitting this state estimate, which is provided by a Kalman filter\footnote{In the literature, a sensor capable of running a Kalman filter in real time is often referred to as smart sensor.}, yields the best performance.

The chapter is organized into six sections. In addition to the current introductory section, we formulate the causal tradeoff problem in Section~\ref{sec2}. We present our main results on characterization of optimal policies in the presence of packet loss and time delay in Sections~\ref{sec3} and \ref{sec4}. We then provide a numerical example related to a spin-stabilized spacecraft in Section~\ref{sec5}. Finally, we conclude the chapter in Section~\ref{sec6}.

\section{A Tradeoff in the Lossy and Delayed Communication Regime}\label{sec2}
In practice, it is often the case that not all states of a dynamical process can be determined by direct observation. State estimation aims at inferring the hidden states of a dynamical process based on the available measurements. This inference is in fact a cornerstone of any subsequent decision-making. In the following, we focus on state estimation of a partially observable process over a lossy and delayed channel by considering simultaneously costs of communication and estimation.

Suppose that time is discretized into time slots over a finite time horizon $T$ such that the duration of each time slot is constant. Consider a sensor observing the states of a partially observable Gauss--Markov process. Accordingly, the source model is expressed by the state and output equations
\begin{align}
	x(k+1) &= A(k) x(k) + w(k) \label{eq:sys}\\[1.5\jot]
	y(k) &= C(k) x(k) + v(k) \label{eq:sens}
\end{align}
for $k \in \mathbb{N}_{[0,T]}$ with initial condition $x(0)$, where $x(k) \in \mathbb{R}^n$ is the state of the source, $A(k) \in \mathbb{R}^{n \times n}$ is the state matrix, $w(k) \in \mathbb{R}^n$ is a Gaussian white noise with zero mean and covariance $W(k) \succ 0$, $y(k) \in \mathbb{R}^m$ is the output of the sensor, $C(k) \in \mathbb{R}^{m \times n}$ is the output matrix, and $v(k) \in \mathbb{R}^m$ is a Gaussian white noise with zero mean and covariance $V(k) \succ 0$. 
\begin{assumption}
For the source model in (\ref{eq:sys}) and (\ref{eq:sens}), the following assumptions are satisfied:
\begin{enumerate}[(i)]
	\item The initial condition $x(0)$ is a Gaussian vector with mean $m(0) \in \mathbb{R}^n$ and covariance $M(0) \succ 0$.
	\item The random variables $x(0)$, $w(t)$, and $v(s)$ for $t,s \in \mathbb{N}_{[0,T]}$ are mutually independent, i.e., $\Prob(x(0), w(0\!:\!T), v(0\!:\!T)) = \Prob(x(0)) \prod_{k=0}^{T} \Prob(w(k))$ $\times \prod_{k=0}^{T} \Prob(v(k))$.
\end{enumerate}
\end{assumption}

The sensor is connected over a fixed-delay packet-erasure channel to a monitor where an estimate of the state of the source, represented by $\hat{x}(k)$, is computed at each time $k$. Associated with this networked system, there are an encoder and a decoder. Let $\sigma(k) \in \{0,1\}$ be a binary variable such that $\sigma(k) = 1$ if a message containing sensory information, represented by $\check{x}(k)$, is transmitted by the encoder to the decoder at time $k$, and $\sigma(k) = 0$ otherwise. A transmitted message at time $k$ is successfully received by the decoder with probability $1 - \lambda(k)$ and after fixed $d$-step delay, where $d \in \mathbb{N}_{[1,T]}$. Let $\gamma(k) \in \{0,1\}$ be a binary random variable modeling packet loss such that $\gamma(k) = 0$ if a packet loss occurs at time $k$, and $\gamma(k) = 1$ otherwise. Then, the probability of $\gamma(k)=0$ is $\lambda(k)$. Accordingly, the channel model is expressed by the input-output relation
\begin{align}\label{eq:ch-model}
z(k+d) = \left\{
  \begin{array}{l l}
     \check{x}(k), & \ \text{if} \ \sigma(k) = 1 \ \wedge \ \gamma(k) =1, \\[1\jot]
     \mathfrak{D}, & \ \text{otherwise}
  \end{array} \right.
\end{align}
for $k \in \mathbb{N}_{[0,T]}$ with $z(0), \dots, z(d-1) = \mathfrak{D}$ by convention, where $z(k)$ is the output of the channel and $\mathfrak{D}$ represents packet loss or absence of transmission. Upon a successful delivery, the decoder transmits a packet acknowledgment to the encoder via a feedback channel.

\begin{assumption}
For the communication channel model in (\ref{eq:ch-model}), the following assumptions are satisfied:
\begin{enumerate}[(i)]
	\item The packet-loss probabilities $\lambda(k)$ for $k \in \mathbb{N}_{[0,T]}$ are random variables forming a Markov chain, i.e., $\Prob(\lambda(k) | \lambda(0:k-1)) = \Prob(\lambda(k) | \lambda(k-1))$.
	\item The packet-loss probability $\lambda(k)$ is known at the encoder and the decoder at each time $k$.
	\item The random variables $\gamma(k)$ for $k \in \mathbb{N}_{[0,T]}$ are mutually independent given the respective packet-loss probabilities, i.e., $\Prob(\gamma(0:T) | \lambda(0:T)) = \prod_{k=0}^{T}$ $\Prob(\gamma(k) | \lambda(k))$.
	\item Quantization error is negligible, i.e., in a successful transmission, real-value sensory information can be conveyed from the encoder to the decoder without any error.
	\item The feedback channel is ideal, i.e., packet acknowledgments are received by the encoder without any loss, delay, or error.
\end{enumerate}
\end{assumption}

\begin{assumption}
The following factorization among the variables associated with the source and the channel holds:
\begin{align}
&\Prob \big(x_0, w(0:T), v(0:T), \gamma(0:T), \lambda(0:T) \big) \nonumber\\[2\jot]
&\qquad = \Prob \big(x_0, w(0:T), v(0:T) \big) \Prob \big(\gamma(0:T) \big| \lambda(0:T) \big)  \Prob \big(\lambda(0:T) \big).
\end{align}
\end{assumption}
\vspace{0.3cm}

In this state estimation scenario, the decision variables are $\sigma(k)$ and $\hat{x}(k)$ for all $k \in \mathbb{N}_{[0,T]}$, which are decided by the encoder and the decoder, respectively. Let the information sets of the encoder and the decoder be expressed by $\mathcal{I}(k) = \{ y(t), z(t), \lambda(t), \sigma(s), \hat{x}(s), \gamma(r) | \ t \in \mathbb{N}_{[0,k]}, s \in \mathbb{N}_{[0,k-1]}, r \in \mathbb{N}_{[0,k-d]} \}$ and $\mathcal{J}(k) = \{ z(t), \lambda(t), \hat{x}(s) \ | \  t \in \mathbb{N}_{[0,k]}, s \in \mathbb{N}_{[0,k-1]} \}$, respectively. We say that a policy profile $(\epsilon,\delta)$ consisting of a scheduling policy $\epsilon$ and an estimation policy $\delta$ is admissible if $\epsilon = \{ \ProbM(\sigma(k) | \mathcal{I}(k)) \}_{k=0}^{T}$ and $\delta = \{ \ProbM(\hat{x}(k) | \mathcal{J}(k)) \}_{k=0}^{T}$, where $\ProbM(\sigma(k) | \mathcal{I}(k))$ and $\ProbM(\hat{x}(k) | \mathcal{J}(k))$ are stochastic kernels. 

We would like to find a globally optimal solution $(\epsilon^\star,\delta^\star)$ to the causal tradeoff between the packet rate and the mean square error, which is cast by the following stochastic optimization problem:
\begin{align}\label{eq:main_problem1}
	&\underset{\epsilon \in \mathcal{E},\delta\in \mathcal{D}}{\minimize} \ \Phi
\end{align}
where $\mathcal{E}$ and $\mathcal{D}$ are the sets of admissible scheduling policies and admissible estimation policies, respectively, and 
\begin{align}\label{eq:loss-function}
\Phi := \EXP \bigg[ \sum_{k=0}^{T} \theta(k) \sigma(k) + \sum_{k=0}^{T} \big( x(k) - \hat{x}(k)\big)^T \Lambda(k) \big( x(k) - \hat{x}(k)\big) \bigg]
\end{align}
for the weighting coefficient $\theta(k) \geq 0$ and the weighting matrix $\Lambda(k) \succ 0$, subject to the source model in (\ref{eq:sys}) and (\ref{eq:sens}) and the channel model in (\ref{eq:ch-model}).

\begin{remark}
Note that the loss function in (\ref{eq:loss-function}), when divided by $T$, incorporates two criteria. The packet-rate criterion, which often appears in the analysis of packet-switching networks, evaluates the cost of communication, while the mean-square-error criterion, which often appears in the analysis of control systems, evaluates the quality of real-time estimation. Moreover, note that the causal tradeoff in~(\ref{eq:main_problem1}) is a team decision-making problem with a non-classical information structure subject to a signaling effect. Although this problem is nonconvex and in general intractable, in the next section we characterize a globally optimal solution $(\epsilon^\star,\delta^\star)$, which leads to the determination of the fundamental performance limit of the underlying networked~system.
\end{remark}

\section{Characterization of an Optimal Solution}\label{sec3}
In this section, we present our theoretical results on the characterization of an optimal solution, which consists of policies for the encoder and the decoder. As we will observe, the optimal design of the encoder and the decoder is generally intertwined. However, we will overcome this hindrance by seeking a separation in the design of these entities. We first show in the next lemma that the conditional means $\EXP[x(k) | \mathcal{I}(k)]$ and $\EXP[x(k) | \mathcal{J}(k)]$ are instrumental for the tasks carried out by the encoder and the decoder.
\begin{lemma}[\hspace{-0.01mm}\cite{erasure2023}]
Without loss of optimality, at each time $k$, one can adopt $\check{x}(k) = \EXP[x(k) | \mathcal{I}(k)]$ as the message that can be transmitted by the encoder, and $\hat{x}(k) = \EXP[x(k) | \mathcal{J}(k)]$ as the state estimate that can be computed by the decoder.
\label{lemma:0}
\end{lemma}

As a result, we can safely use $\check{x}(k) = \EXP[x(k) | \mathcal{I}(k)]$ and $\hat{x}(k) = \EXP[x(k) | \mathcal{J}(k)]$. Define the innovation at the encoder $\nu(k) := y(k) - C(k) \EXP [x(k) | \mathcal{I}(k-1)]$, the estimation error at the decoder $\hat{e}(k) := x(k) - \hat{x}(k)$, and the estimation mismatch $\tilde{e}(k) := \check{x}(k) - \hat{x}(k)$. We characterize in the next two lemmas the recursive equations that the optimal estimators at the encoder and the decoder must~satisfy.

\begin{lemma}[\hspace{-0.01mm}\cite{stoccontrol}]
The minimum mean-square-error estimator at the encoder satisfies the recursive~equations
\begin{align}
	\check{x}(k) &= m(k) + K(k) \big( y(k) - C(k) m(k) \big) \label{eq:est-KF-xhat} \\[2\jot]
	m(k) &= A(k-1) \check{x}(k-1) \label{eq:est-KF-m}\\[1\jot]
	O(k) &= \big( M(k)^{-1} + C(k)^T V(k)^{-1} C(k) \big)^{-1}\\[1\jot]
	M(k) &= A(k-1) O(k-1) A(k-1)^T + W(k-1)
\end{align}
for $k \in \mathbb{N}_{[1,T]}$ with initial conditions $\check{x}(0) = m(0) + K(0)(y(0) - C(0) m(0))$ and $O(0) = (M(0)^{-1} + C(0)^T V(0)^{-1} C(0))^{-1}$, where $m(k) = \EXP[ x(k) | \mathcal{I}(k-1)]$, $O(k) = \Cov[x(k) | \mathcal{I}(k)]$, $M(k) = \Cov[x(k) | \mathcal{I}(k-1)]$, and $K(k) = O(k) C(k)^T V(k)^{-1}$.
\label{lemma:1}
\end{lemma}

\begin{lemma}[\hspace{-0.01mm}\cite{erasure2023}]
The minimum mean-square-error estimator at the decoder satisfies the recursive equation
\begin{align}\label{est-monitor}
	&\hat{x}(k) =  \sigma(k-d) \gamma(k-d) \bigg( \prod_{t=1}^{d} A(k-t) \bigg) \check{x}(k-d) \nonumber\\[0.5\jot]
	&\qquad \ \ \ + \big( 1-\sigma(k-d) \gamma(k-d) \big) \big( A(k-1) \hat{x}(k-1) + \varsigma(k-1) \big)
\end{align}
for $k \in \mathbb{N}_{[d,T]}$ with initial conditions $\hat{x}(\tau) = (\prod_{t=1}^{\tau} A(\tau-t)) m(0)$ for $\tau \in \mathbb{N}_{[0,d-1]}$, where $\varsigma(k-1) = A(k-1) \EXP[\hat{e}(k-1) | \mathcal{J}(k-1), \sigma(k-d) \gamma(k-d) = 0]$ is a signaling residual.
\label{lemma:x2}
\end{lemma}

\begin{remark}
The results of Lemmas~\ref{lemma:1} and \ref{lemma:x2} have shown that the conditional distribution $\ProbM({x}(k) | \mathcal{I}(k))$ is Gaussian and the optimal estimator at the encoder is linear, while the conditional distribution $\ProbM({x}(k) | \mathcal{J}(k))$ is in general non-Gaussian and the optimal estimator at the decoder is in general nonlinear. The nonlinearity of the optimal estimator at the decoder is in fact caused by the signaling residual $\varsigma(k)$, which is defined when $z(k)$ is equal to $\mathfrak{D}$. The existence of this term implies that the decoder might be able to decrease its uncertainty even when no message is transmitted or when a packet loss occurs. Finally, it is worth mentioning that $\mathcal{J}(k) \subset \mathcal{I}(k)$. Therefore, both $\check{x}(k)$ and $\hat{x}(k)$ are known at the encoder at each time $k$.
\end{remark}

We introduce in the next definition a value function by adopting the estimation policy $\delta^\star$ that is constructed based on $\EXP[x(k) | \mathcal{J}(k)]$.
\begin{definition}[Value function]
The value function $V(k,\mathcal{I}(k))$ associated with the loss function $\Phi$ and the information set $\mathcal{I}(k)$ is defined as
\begin{align}
	V(k,\mathcal{I}(k)) :=& \min_{\epsilon \in \mathcal{E} : \delta = \delta^\star}\EXP \bigg[ \sum_{t=k}^{T-d} \theta(t) \sigma(t) + \sum_{t=k}^{T-d} \hat{e}(t+d)^T \hat{e}(t+d) \Big| \mathcal{I}(k) \bigg] \label{eq:Ve-def}
\end{align}
for $k \in \mathbb{N}_{[0,T]}$.
\label{def:valuefunc}
\end{definition}

Moreover, analogous to the findings in~\cite{voi, voi2, touraj-thesis}, we introduce in the next definition the general formula of the value of information.
\begin{definition}[Value of Information]
The value of information at time $k$ is defined as the variation in the value function $V(k,\mathcal{I}(k))$ with respect to the sensory measurement $\check{x}(k)$ that can be communicated to the decoder at time~$k$, i.e.,
\begin{align}\label{eq:voi-def}
\voi(k,\mathcal{I}(k)) := V(k,\mathcal{I}(k))|_{\sigma(k) = 0} - V(k,\mathcal{I}(k))|_{\sigma(k) = 1}
\end{align}
where $V(k,\mathcal{I}(k))|_{\sigma(k)}$ denotes the value function $V(k,\mathcal{I}(k))$ when the transmission decision $\sigma(k)$ is enforced.
\label{def:voi}
\end{definition}

Let us define the variables $\chi(k) := \EXP[ \tilde{e}(k+d)^T \Lambda(k+d) \tilde{e}(k+d) + V(k+1,\mathcal{I}(k+1))| \mathcal{I}(k), \sigma(k) = 0] - \EXP[ \tilde{e}(k+d)^T \Lambda(k+d) \tilde{e}(k+d) + V(k+1,\mathcal{I}(k+1))| \mathcal{I}(k),\sigma(k) = 1]$. In the next theorem, we characterize a globally optimal solution in the causal tradeoff between the packet rate and the mean square error.
\begin{theorem}[\hspace{-0.01mm}\cite{erasure2023}]
The causal tradeoff between the packet rate and the mean square error in (\ref{eq:main_problem1}) pertaining to state estimation of a partially observable process modeled by (\ref{eq:sys}) and (\ref{eq:sens}) over a lossy and delayed channel modeled by (\ref{eq:ch-model}) admits a globally optimal solution $(\epsilon^{\star},\delta^{\star})$ such that $\epsilon^{\star}$ is a symmetric threshold scheduling policy given~by
\begin{align}\label{eq:opt-mac-policy}
	\sigma(k) = \mathds{1}_{ \voi(k,\mathcal{I}(k)) \geq 0}
\end{align}
for $k \in \mathbb{N}_{[0,T-d]}$, where $\voi(k,\mathcal{I}(k)) =  \chi(k) - \theta(k)$ is a symmetric function of $\nu(0\!:\!k)$ and expressible in terms of $\tilde{e}(k)$, $\nu(k-d+2\!:\!k)$, $\lambda(k-d+1\!:\!k)$, and $\sigma(k-d+1\!:\!k-1)$; and $\delta^\star$ is a non-Gaussian linear estimation policy given~by
\begin{align}\label{eq:opt-est-policy}
	&\hat{x}(k) =  \sigma(k-d) \gamma(k-d) \bigg( \prod_{t=1}^{d} A(k-t) \bigg) \check{x}(k-d) \nonumber\\[0.5\jot]
	&\qquad \qquad + \big(1 - \sigma(k-d) \gamma(k-d) \big) A(k-1) \hat{x}(k-1)
\end{align}
for $k \in \mathbb{N}_{[d,T]}$ without being influenced by the signaling residual $\varsigma(k-1)$, with initial conditions $\hat{x}(\tau) = (\prod_{t=1}^{\tau} A(\tau-t) ) m(0)$ for $\tau \in \mathbb{N}_{[0,d-1]}$.
\label{thm:1}
\end{theorem}

\begin{remark}
The results of Theorem~\ref{thm:1} have shown the existence of a globally optimal solution that is composed of a symmetric threshold scheduling policy and a non-Gaussian linear estimation policy. First, note that the scheduling policy is expressible in terms of $3d-1$ variables related to the source and the channel (i.e., $\tilde{e}(k)$, $\nu(k-d+2\!:\!k)$, $\lambda(k-d+1\!:\!k)$, and $\sigma(k-d+1\!:\!k-1)$ at each time $k$), and that the estimation policy incorporates no signaling residual (i.e., $\varsigma(k)$ becomes equal to zero for all $k \in \mathbb{N}_{[d-1,T-1]}$). The former implies that the encoder does not require a large memory to save the whole history, and the latter implies that the decoder need not perform complex computations. Moreover, note that these scheduling and estimation policies can be designed separately, while we have already observed from (\ref{est-monitor}) and (\ref{eq:Ve-def}) that the optimal design of the encoder and the decoder is generally intertwined as $\epsilon$ depends on $\delta$ and vice versa. Finally, note that, at the characterized policy profile, the benefit in transmitting a message is equal to $\chi(k)$ and its cost is equal to $\theta(k)$. Accordingly, a message is transmitted from the encoder to the decoder at each time only when its benefit surpasses its cost. This event-triggered property of the networked system based on the value of information has directly emerged from our cost-benefit analysis.
\end{remark}

\section{Transition from State Estimation to Feedback Control}\label{sec4}
We now pivot towards feedback control in the state space, which leverages knowledge of the process's internal states in order to regulate its behavior. Clearly, the accuracy of state estimation directly influences the efficacy of feedback control. Therefore, a transition from the former to the latter in fact requires a synergy between information and action, which is crucial for achieving optimal performance in the presence of uncertainties. In the following, we focus on feedback control of a partially observable process over a lossy and delayed channel by considering simultaneously costs of communication and regulation.

Suppose that the dynamical process discussed in Section (\ref{sec2}) now needs to be regulated. Accordingly, the source model is updated as
\begin{align}
	x(k+1) &= A(k) x(k) + B(k) u(k) + w(k)\label{eq:sys2}\\[2\jot]
	y(k) &= C(k) x(k) + v(k) \label{eq:sens2}
\end{align}
for $k \in \mathbb{N}_{[0,T]}$ with initial condition $x(0)$, where $B(k) \in \mathbb{R}^{n \times m}$ is the input matrix and $u(k) \in \mathbb{R}^m$ is the actuation input. The channel model is defined as before, with the input-output relation in (\ref{eq:ch-model}).

In this feedback control scenario, the decision variables are $\sigma(k)$ and $u(k)$ for all $k \in \mathbb{N}_{[0,T]}$, which are decided by the encoder and the decoder, respectively. Let the information sets of the encoder and the decoder be expressed by $\mathcal{I}(k) = \{ y(t), z(t), \lambda(t), \sigma(s), u(s), \gamma(r) | \ t$ $\in \mathbb{N}_{[0,k]}, s \in \mathbb{N}_{[0,k-1]}, r \in \mathbb{N}_{[0,k-d]} \}$ and $\mathcal{J}(k) = \{ z(t), \lambda(t), u(s) \ | \  t \in \mathbb{N}_{[0,k]}, s \in \mathbb{N}_{[0,k-1]} \}$, respectively. We say that a policy profile $(\epsilon,\delta)$ consisting of a scheduling policy $\epsilon$ and a control policy $\delta$ is admissible if $\epsilon = \{ \ProbM(\sigma(k) | \mathcal{I}(k)) \}_{k=0}^{T}$ and $\delta = \{ \ProbM(u(k) | \mathcal{J}(k)) \}_{k=0}^{T}$, where $\ProbM(\sigma(k) | \mathcal{I}(k))$ and $\ProbM(u(k) | \mathcal{J}(k))$ are stochastic kernels.

We would like to find a globally optimal solution $(\epsilon^\star,\delta^\star)$ to the causal tradeoff between the packet rate and the regulation cost, which is cast by the following stochastic optimization problem:
\begin{align}\label{eq:main_problem2}
	&\underset{\epsilon \in \mathcal{E},\delta\in \mathcal{D}}{\minimize} \ \Phi'
\end{align}
where now $\mathcal{E}$ and $\mathcal{D}$ are the sets of admissible scheduling policies and admissible control policies, respectively, and
\begin{align}\label{eq:loss-function2}
\Phi' := \EXP \bigg[ \sum_{k=0}^{T} \theta(k) \sigma(k) + \sum_{k=0}^{T+1} x(k)^T Q(k) x(k) + \sum_{k=0}^{T} u(k)^T R(k) u(k) \bigg]
\end{align}
for the weighting coefficient $\theta(k) \geq 0$ and the weighting matrices $Q(k) \succeq 0$ and $R(k) \succ 0$, subject to the source model in (\ref{eq:sys2}) and (\ref{eq:sens2}) and the channel model in (\ref{eq:ch-model}).

The next lemma introduces an equivalent loss function based on an algebraic Riccati equation.
\begin{lemma}[\hspace{-0.01mm}\cite{stoccontrol}]
Let $S(k) \succeq 0$ be a matrix obeying the algebraic Riccati equation
\begin{equation}\label{eq:riccati}
\begin{aligned}
S(k) &= Q(k) + A(k)^T S(k+1) A(k) - A(k)^T S(k+1) B(k)\\[2.5\jot]
	&\qquad \quad \times (B(k)^T S(k+1) B(k) + R(k))^{-1} B(k)^T S(k+1) A(k)
\end{aligned}
\end{equation}
for $k \in \mathbb{N}_{[0,T]}$ with initial condition $S(T+1) = Q(T+1)$. Then,
\begin{align}\label{eq:psi}
	\Psi := \EXP \Big[ \sum_{k=0}^{T} \theta(k) \sigma(k) + \eta(k) \Big]
\end{align}
is equivalent to $\Phi$, where $\eta(k) =  (u(k) + L(k) x(k))^T (B(k)^T S(k+1) B(k) + R(k)) (u(k) + L(k) x(k))$ and $L(k) = (B(k)^T S(k+1) B(k) + R(k))^{-1} B(k)^T S(k+1) A(k)$.
\label{lem:4}
\end{lemma}

Define the variables $n(k) := B(k-1) u(k-1)$, $r(k) := \sigma(k-d) \gamma(k-d) \sum_{t=1}^{d-1} ( \prod_{t'=1}^{t-1} A(k-t') ) B(k-t) u(k-t) + ( 1 - \sigma(k-d) \gamma(k-d)) B(k-1) u(k-1)$, and $\Gamma(k) := A(k)^T S(k+1) B(k) (B(k)^T S(k+1) B(k) + R(k))^{-1} B(k)^T S(k+1) A(k)$. In the next theorem, we characterize a globally optimal solution in the causal tradeoff between the packet rate and the regulation cost.
\begin{theorem}
The causal tradeoff between the packet rate and the regulation cost in (\ref{eq:main_problem2}) pertaining to feedback control of a partially observable process modeled by (\ref{eq:sys2}) and (\ref{eq:sens2}) over a lossy and delayed channel modeled by (\ref{eq:ch-model}) admits a globally optimal solution $(\epsilon^{\star},\delta^{\star})$ such that $\epsilon^{\star}$ is a symmetric threshold scheduling policy given~by (14) when $m(k) = n(k)$ and $\Lambda(k) = \Gamma(k)$; and $\delta^{\star}$ is a certainty-equivalent control policy given by
\begin{align}\label{eq:optimal-control}
u(k) = - L(k) \hat{x}(k)
\end{align}
for $k \in \mathbb{N}_{[0,T]}$, where
\begin{align}
	&\hat{x}(k) = \sigma(k-d) \gamma(k-d) \bigg(\prod_{t=1}^{d} A(k-t) \bigg) \check{x}(k-d)\nonumber\\[0.5\jot]
	&\qquad \qquad + \big( 1 - \sigma(k-d) \gamma(k-d) \big) A(k-1) \hat{x}(k-1) \nonumber\\[0.5\jot]
		&\qquad \qquad + \sigma(k-d) \gamma(k-d) \sum_{t=1}^{d} \bigg( \prod_{t'=1}^{t-1} A(k-t') \bigg) B(k-t) u(k-t) \nonumber\\[0.5\jot]
	&\qquad \qquad + \big( 1 - \sigma(k-d) \gamma(k-d) \big) B(k-1) u(k-1)
\end{align}
for $k \in \mathbb{N}_{[d,T]}$ without being influenced by the signaling residual $\varsigma(k-1)$, with initial conditions $\hat{x}(\tau) = (\prod_{t=1}^{\tau} A(\tau-t) ) m(0) + \sum_{t=1}^{\tau}(\prod_{t'=1}^{t-1} A(\tau - t')) B(\tau - t) u(\tau-t)$ for $\tau \in \mathbb{N}_{[0,d-1]}$.
\label{thm:2}
\end{theorem}
\begin{proof}
The derivation of the control policy follows directly from the analysis in \cite{voi2}, where it was shown that the separation principle holds. Inserting the actuation input $u(k) = - L(k) \hat{x}(k)$ into (\ref{eq:psi}), we get $\eta(k) = (x(k) - \hat{x}(k))^T \Gamma(k) (x(k) - \hat{x}(k))$. This implies that $\Psi$ in Lemma~\ref{lem:4}, which is equivalent to $\Phi'$, will be the same as $\Phi$ if $\Lambda(k) = \Gamma(k)$. Moreover, the associated minimum mean-square-error estimators at the encoder and the decoder are obtained according to Lemmas \ref{lemma:1} and \ref{lemma:x2} when the previously applied actuation inputs are incorporated. This means adding $n(k)$ to (\ref{eq:est-KF-m}), adding $r(k)$ to (\ref{est-monitor}), and updating the initial conditions $\hat{x}(\tau) = (\prod_{t=1}^{\tau} A(\tau-t) ) m(0) + \sum_{t=1}^{\tau}(\prod_{t'=1}^{t-1} A(\tau - t')) B(\tau - t) u(\tau-t)$ for $\tau \in \mathbb{N}_{[0,d-1]}$. Finally, the derivation of the optimal scheduling policy follows from the results of Theorem~\ref{thm:1} when $\Lambda(k) = \Gamma(k)$.
\end{proof}

\begin{remark}
The results of Theorem \ref{thm:2} have shown that the optimal control policy is a certainty-equivalent policy with a switching filter equation. Note that this control policy assumes that uncertainties can be replaced by their expected values, essentially converting the stochastic system into a deterministic one. Such a property of course offers simplifications and computational advantages. Moreover, it is important to note that the optimal scheduling policy has remained structurally the same as the one proposed earlier in Theorem~\ref{thm:1}. In particular, this scheduling policy is still expressible in terms of $3d-1$ variables related to the source and the channel (i.e., $\tilde{e}(k)$, $\nu(k-d+2\!:\!k)$, $\lambda(k-d+1\!:\!k)$, and $\sigma(k-d+1\!:\!k-1)$ at each time $k$).
\end{remark}

\begin{figure}[t!]
\centering
  \includegraphics[width=.83\linewidth]{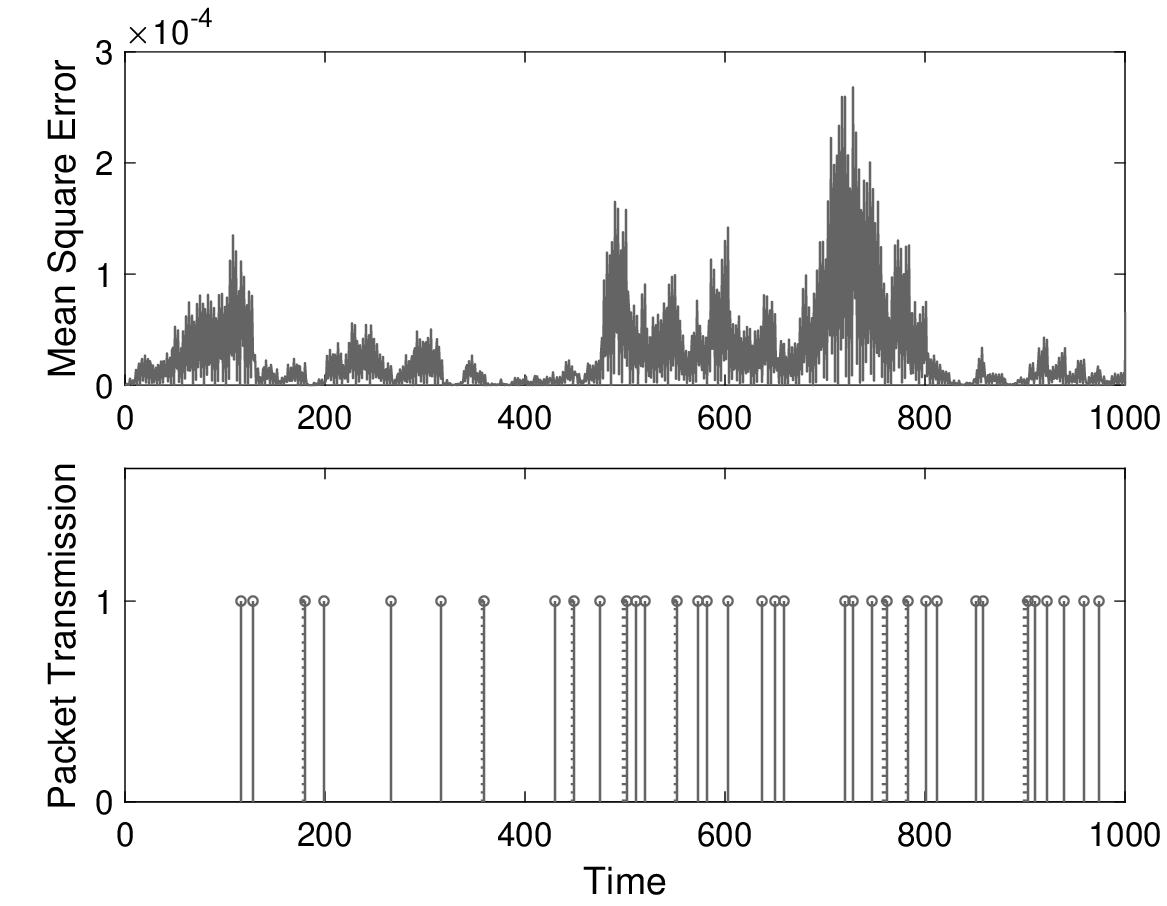}
  \caption{Mean square error and packet transmission trajectories under the optimal policy profile. The total mean square error is $0.0536$, and the total number of packet transmissions is $46$, out of which $11$ were lost. The solid lines represent successful deliveries, and the dotted lines represent packet losses.}
  \label{fig:trajec-voi}
\end{figure}

\begin{figure}[t!]
\centering
  \includegraphics[width=.83\linewidth]{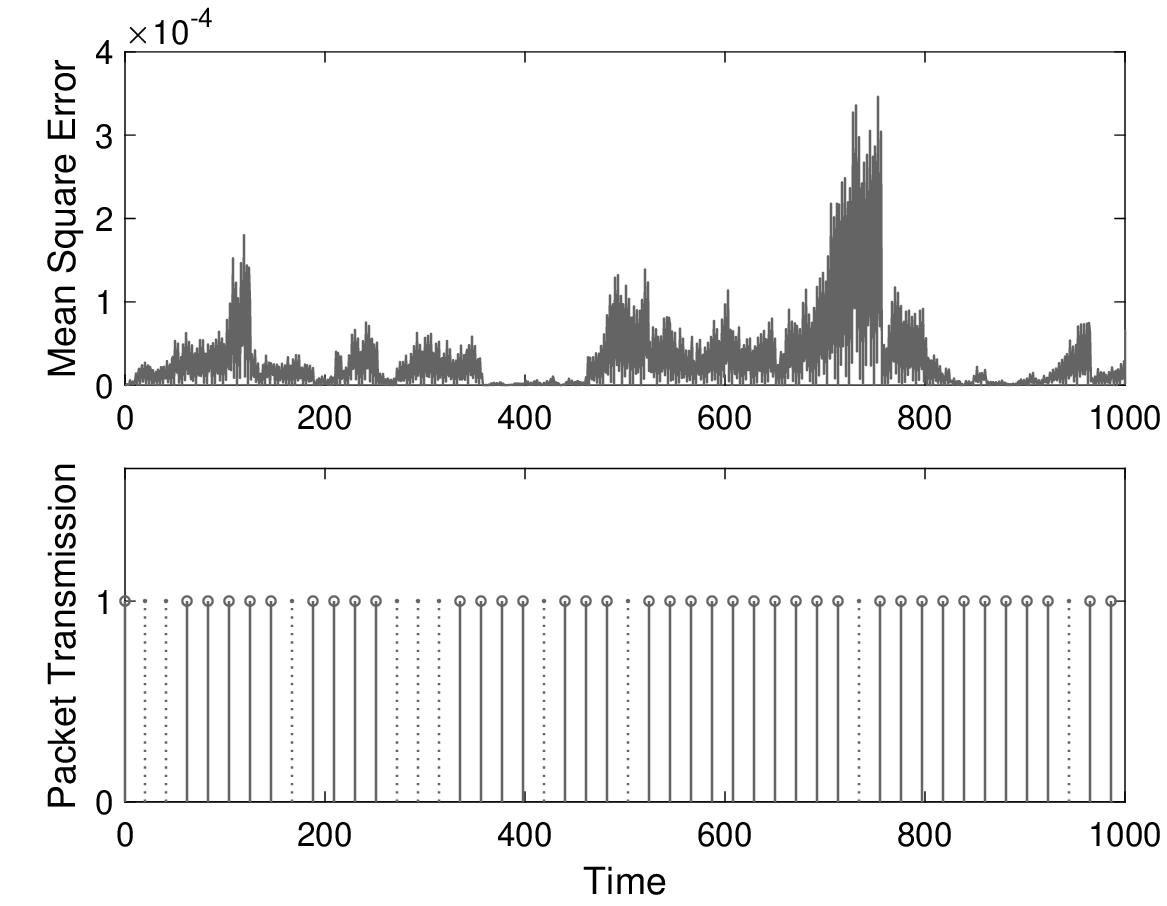}
  \caption{Mean square error and packet transmission trajectories under a periodic policy profile. The total mean square error is $0.0610$, and the total number of packet transmissions is $48$, out of which $10$ were lost. The solid lines represent successful deliveries, and the dotted lines represent packet losses.}
  \label{fig:trajec-per}
\end{figure}

\section{Numerical Example}\label{sec5}
In this section, we provide a numerical example to demonstrate how the framework developed in the previous sections works. Consider a spin-stabilized spacecraft whose body is spinning about the $z$-axis, i.e., the axis of symmetry, with a constant angular velocity $\omega_z = \omega_0$. For such a vehicle, the Euler equation is written as
\begin{align*}\setlength\arraycolsep{6pt}\def\arraystretch{1.5}
\begin{bmatrix}
 \dot{\omega}_x\\
 \dot{\omega}_y\\
 \dot{\omega}_z     
 \end{bmatrix}
 =
\begin{bmatrix}
 0 & \frac{I_y - I_z}{I_x} \omega_0 & 0\\
 \frac{I_z - I_x}{I_y} \omega_0 & 0 & 0 \\
 0 & 0 & 0   
 \end{bmatrix}
\begin{bmatrix}
 \omega_x\\
 \omega_y\\
 \omega_z     
 \end{bmatrix} 
+
\begin{bmatrix}
e_x\\
e_y \\
e_z	
\end{bmatrix}
\end{align*}
where $(\omega_x,\omega_y,\omega_z)$ is the angular velocity, $(I_x,I_y,I_z)$ is the moment of inertia, and $(e_x,e_y,e_z)$ is a Gaussian disturbance torque acting on the spacecraft. Note that for spin stability, the spacecraft must be spinning either about the major or minor axis of inertia. In our example, the following parameters are used: $I_x = I_y = 20 \ \text{kg.m$^2$}$, $I_z = 100 \ \text{kg.m$^2$}$, and $\omega_0 = 2 \epsilon \ \text{rad/s}$; and the Euler equation is discretized over the time horizon $T = 1000$. The state equation of the form (\ref{eq:sys}) is specified by $A(k) = [0.4258, 0.4258, 0; 0.4258, 0.4258, 0; 0, 0 , 1]$ and $W(k) = 10^{-6} \diag\{0.2245,0.2245,0.0025\}$ for all $k \in \mathbb{N}_{[0,T]}$ with $m(0) = [0;0;2\pi]$ and $M(0) = 10W(k)$. Suppose there is a sensor on the spacecraft that partially observes each component of the angular velocity at each time $k$. The output equation of the form (\ref{eq:sens}) is specified by $C(k) = \diag \{1,1,1\}$ and $V(k) = 10^{-3} \diag \{ 1, 1, 1\}$ for all $k \in \mathbb{N}_{[0,T]}$. The sensory measurements should be transmitted over a costly downlink channel to a ground station where the angular velocity is estimated. The downlink channel is subject to packet loss with $\lambda(k) = 0.3$ for all $k \in \mathbb{N}_{[0,T]}$ and time delay with $d = 2$.

For this networked system, we are interested in finding the optimal policy profile $(\epsilon^\star,\delta^\star)$ that minimizes the loss function $\Phi$ in the causal tradeoff between the packet rate and the mean square error with the weighting coefficient $\theta(k) = 8 \times 10^{-6}$ for $k \in \mathbb{N}_{[0,T/2]}$ and $\theta(k) = 6 \times 10^{-6}$ for $k \in \mathbb{N}_{[T/2+1,T]}$, and the weighting matrix $\Lambda(k) = \diag \{ 1, 1, 1\}$ for all $k \in \mathbb{N}_{[0,T]}$. For a simulated realization of the system, Fig.~\ref{fig:trajec-voi} shows the mean square error and packet transmission trajectories under the optimal policy profile. In this case, the total mean square error is $0.0536$, and the total number of packet transmissions is $46$, out of which $11$ were lost. Besides, for the same realization of the system, Fig.~\ref{fig:trajec-per} shows the corresponding trajectories under a periodic policy profile. In this case, the total mean square error is $0.0610$, and the total number of packet transmissions is $48$, out of which $10$ were lost. We observe in this example that the optimal policy profile proved effective in improving the system performance. It is interesting to note that, in comparison with the periodic scheduling policy, the optimal scheduling policy not only transmits sensory information less frequently when the estimation discrepancy is small, but transmits more frequently and more persistently when the estimation discrepancy is large due to packet losses and time delay. Moreover, the optimal scheduling policy is adaptive to the weighting coefficient $\theta(k)$, and transmits more frequently when the cost of communication~decreases.

\section{Conclusions}\label{sec6}
In this chapter, we first studied the fundamental performance limit of state estimation of a partially observable process over a lossy and delayed channel. To that end, we formulated a causal tradeoff between the packet rate and the mean square error. Associated with this tradeoff, we characterized a globally optimal policy profile, and showed that this policy profile is composed of a symmetric threshold scheduling policy based on the value of information and a non-Gaussian linear estimation policy. We then extended the results to feedback control, and showed that the optimal control policy is a certainty-equivalent policy that exploits the derived estimation policy. Leveraging these results and quantifying the value of information according to the prior works, we have been able to validate the consistency of the value of information in a lossy and delayed communication regime. This comprehensive analysis contributes to our understanding of the interplay between information, estimation, and control in networked control systems operating under challenging communication conditions.

\bibliography{../../../../mybib}
\bibliographystyle{ieeetr}

\end{document}